\documentclass{article}
\usepackage{fullpage}
\usepackage{epsfig}
\usepackage{graphics}
\usepackage{latexsym}
\usepackage{amsmath}
\usepackage{amsfonts}
\usepackage{amssymb}
\usepackage[mathscr]{eucal}
\usepackage{mathrsfs}
\usepackage{pifont}
\usepackage{yhmath}
\usepackage{undertilde}
\usepackage[usenames]{color}
\usepackage{amsthm}

\usepackage{url}
\urldef{\mailsa}\path|{zhigangcao,xchen,wcj}@amss.ac.cn|

\newtheorem{theorem}{Theorem}[section]
\newtheorem{lemma}[theorem]{Lemma}

\newtheorem{algorithm}{Algorithm}{\itshape}{\rmfamily}

\newtheorem{example}[theorem]{Example}

 \theoremstyle{remark}

\newtheorem{clm}{Claim}{\itshape}{\rmfamily}

{\itshape}{\rmfamily}

\newtheorem{remark}[theorem]{Remark}

\newtheorem{observation}{Observation}{\itshape}{\rmfamily}

\newtheorem{procedure}{Procedure}{\itshape}{\rmfamily}

\makeatletter
\@addtoreset{equation}{section}
\def\section{\@startsection {section}{1}{\z@}{-3.5ex plus -1ex minus
 -.2ex}{2.3ex plus .2ex}{\large\bf}}
\makeatother

\def\bfm#1{\mbox{\boldmath$#1$}}

\def\0{\bfm 0}

\newcommand{\comment}[1]{\textcolor{blue}{\textbf{#1}}}

\DeclareMathAlphabet{\mathpzc}{OT1}{pzc}{m}{it}

\def\bfm#1{\mbox{\boldmath$#1$}}

\begin{document}


\title{\bf How to Schedule the Marketing of Products with\\ Negative Externalities \thanks{Supported in part by NNSF of China under Grant No. 11222109, 11021161, 10928102{ and 71101140}, 973 Project of China under Grant No. 2011CB80800 {and 2010CB731405}, and CAS under
Grant No. kjcx-yw-s7.}}


 \author{Zhigang Cao  \and Xujin Chen
 \and
 Changjun Wang}
\date{Academy of Mathematics and Systems Science \\ Chinese Academy
of Sciences, Beijing 100190, China\\
${}$\\
\mailsa}

\maketitle

\begin{abstract} In marketing  products with negative externalities, a schedule which specifies an order of consumer  purchase decisions  is crucial, since  in  the  social network of consumers, the decision of each consumer is  negatively affected by the choices of her neighbors.
In this paper, we study the problems of finding a marketing schedule for     two asymmetric products with negative externalites. The goals  are
 two-fold:
 maximizing the sale  of one product and
  ensuring regret-free purchase decisions.
 We show that the maximization is NP-hard, and
 provide efficient algorithms with satisfactory performance guarantees.  Two of these algorithms give regret-proof schedules, i.e. they reach   Nash equilibria where no consumers regret   their previous decisions.
Our work is the first attempt to address these marketing problems from an algorithmic point of view.
\end{abstract}


\newcounter{my}
\newenvironment{mylabel}
{
\begin{list}{(\roman{my})}{
\setlength{\parsep}{-0mm}
\setlength{\labelwidth}{8mm}
\setlength{\leftmargin}{8mm}
\usecounter{my}}
}{\end{list}}

\newcounter{my2}
\newenvironment{mylabel2}
{
\begin{list}{(\alph{my2})}{
\setlength{\parsep}{-1mm} \setlength{\labelwidth}{12mm}
\setlength{\leftmargin}{14mm}
\usecounter{my2}}
}{\end{list}}

\newcounter{my3}
\newenvironment{mylabel3}
{
\begin{list}{(\alph{my3})}{
\setlength{\parsep}{-1mm}
\setlength{\labelwidth}{8mm}
\setlength{\leftmargin}{10mm}
\usecounter{my3}}
}{\end{list}}

\section{Introduction}
The total value of any (consumer) product can be roughly classified into three parts: physical value, emotional value, and social value \cite{vn10}. With the fast development of economy, the basic physical needs of more and more consumers are easily met. Consequently, people increasingly
shift their attention  to emotional and social values when they consider whether to buy a product.  In particular, the social value, whose amount is not
determined by what a consumer consumes alone or how she personally enjoys it, but by the comparisons with what other people around her consume, is becoming a
more and more crucial ingredient for both consumer purchase and therefore seller marketing. For many products, whether they will be welcome depends mainly on how much social value they can provide to the
consumers. This is especially true for fashionable and luxury  goods, where the products often exhibit {\em negative} ({\em consumption}) {\em externalities} -- they become less valuable as more people use them \cite{a05,hs00}.

 The comparison that a consumer makes, for calculating the social value of a product, is naturally restricted to her neighbors in the  social network. 
 For a consumer, the social value of a product with negative externality 
is often proportional to the number of her 
neighbors who do not consume this product \cite{hs00}. In the market, the purchase decisions of a consumer 
often depend  on the values of the products at the time they are  promoted -- the product of larger value will be selected. In contrast to the physical and emotional values, which are relatively fixed, 
the social values of products vary with different marketing schedules. 
The goal of this paper is to design   good
marketing schedules for promoting products with negative externalities in social networks. 


\bigskip\noindent{\bf Motivation and related work}\quad
Our study is motivated by the practical marketing problem concerning how to bring the products to consumers' attention over time. 
Among 
a large literature on diffusion of competing products  or opinions in social networks (see e.g., \cite{am11,ckp12,gk12} and references therein), 
Chierichetti, Kleinberg and Pancones  \cite{ckp12} recently studied the scheduling aspect of the diffusion problem on two products 
 -- finding an order of  consumer purchase decision making  to maximize the adoption of one product. In their model, the two competing  products both have positive (consumption) externalities and every consumer follows the majority of her social network neighbors
 when the externalities outweigh  her own internal preference. The authors \cite{ckp12} provided an algorithm that ensures an expected linear number of favorable decisions.

 The network-related consumption externalities have been classified into four categories \cite{ct05}. Comparing to the other three, the negative cross-consumer externality, as considered in this paper, has been far less studied \cite{a05,hs00}, and was emphasized for its importance in management and marketing nowdays \cite{ct05}.

The model studied in this paper can also be taken as an extension of one side of {\em the fashion game}, which was formulated by Jackson \cite{j08}. Very interestingly, people often have quite different, in fact almost opposite, opinions on what is fashionable, e.g., ``Lady Gaga is Godness of fashion''  vs ``This year's fashion color is black''. 
Following Jackson, 
 we call consumers 
holding the former ``personality reflection''  idea of fashion {\em rebels} and the latter ``prevailing style'' idea {\em conformists}.  More generally, a consumer behaves like a rebel (conformist) if the product, from her point of view, has negative (positive) externality.
In an era emphasizing personal identities, more and more consumers
would like to be rebels. For example, they would prefer Asian-style pants, when seeing many friends and colleagues (their social network neighbors)
wearing European-style.  However, the rebel social network is still under-researched in comparison with vast literature on conformist social networks. 
For a market where all the consumers are rebels, as considered in this paper, it has been previously studied by several papers under the term
of {anti-coordination} \cite{b07,cy12}. 

\bigskip\noindent{\bf Model formulation}\quad
The market is represented by  a social network $G=(V,E)$, an undirected graph with node set $V$ consisting of  $n$ consumers and link set $E$  of $m$   connections between consumers.
  A seller has two (types of) products $\mathcal Y$ and $\mathcal N$ with similar functions.
    We abuse notations by using $\mathcal Y$ and $\mathcal N$ to denote both   types and   products.

 The marketing is done sequentially: 
The seller 
is able to ask the   consumers
 one by one whether they are more interested in $\mathcal Y$ or in $\mathcal N$. Each consumer buys (chooses) exactly one of $\mathcal Y$ and $\mathcal N$, whichever provides her a larger total value, only at the time  she is asked. {This is a simplification of the so called {\em precision marketing} \cite{zb04}.}
 For every  consumer, a product of type $\mathcal T\in\{\mathcal Y,\mathcal N\}$ provides  her   with total value $p_{{}_{\mathcal T}}+s_{{}_{\mathcal T}}(x_{{}_{\mathcal T}})$, where   $p_{{}_{\mathcal T}}$ is the sum of physical and emotional values, and $s_{{}_{\mathcal T}}(x_{{}_{\mathcal T}})$ is the social value   determined by  decreasing function $s_{{}_{\mathcal T}}(\cdot)$ and the number $x_{{}_{\mathcal T}}$ of her neighbors who have bought product $\mathcal T$. We assume that $\mathcal Y$ is very similar to $\mathcal N$ with $p_{{}_{\mathcal Y}}>p_{{}_{\mathcal N}}$ and the externality outweighs the physical and emotional difference, i.e., for any permutation $\mathcal T,\mathcal F$ of $\mathcal Y,\mathcal N$ and any nonnegative integers $x,y$ ($x<y$)  we have $  s_{{}_{\mathcal T}}(x)-s_{{}_{\mathcal F}}(x)<p_{{}_{\mathcal Y}}-p_{{}_{\mathcal N}}< s_{{}_{\mathcal T}}(x)-s_{{}_{\mathcal F}}(y)$.

Actually, the above model can be summarized as  the following scheduling problems on rebel social networks.

\medskip
\noindent{\em {Rebels.}}\quad {Every consumer is a rebel  who,}
at her turn to choose from $\{\mathcal Y,\mathcal N\}$, will buy the product different from the one currently possessed
by the majority of her neighbors. If there are equal numbers of neighbors having bought $\mathcal Y$ and $\mathcal N$ respectively, the consumer will
always buy $\mathcal Y$.

\medskip
\noindent{\em Scheduling.}\quad  A (marketing)
{\em schedule}  $\pi$, 
for network $G$ is an ordering  of consumers in $V$  which specifies the order $\pi(v)\in\{1,2,\ldots,n\}$ of consumer $v\in V$ being asked to buy (choose) $\mathcal Y$ or $\mathcal N$, or ``being scheduled'' for short. {We refer to the problem of finding a schedule for a rebel social network as the {\em rebel scheduling problem}}. Given schedule $\pi$, the choice (purchase {\em decision}) of each consumer $v$ under $\pi$ is uniquely determined, and we denote it by $\pi[v]$, which belongs to $\{\mathcal Y,\mathcal N\}$.  The decisions of all consumers form the marketing {\em outcome} $(\pi[v]:v\in V)$ of $\pi$.
{The basic goal of the rebel scheduling problem} is to  find a   schedule whose outcome contains    $\mathcal Y$ (resp.  $\mathcal N$) decisions as many as possible because $\mathcal Y$ (resp. $\mathcal N$) is more profitable for the seller. 

\medskip
\noindent{\em Equilibrium.}\quad As seen above, the  value of a product changes as the marketing proceeds.  {Every schedule corresponds to  a dynamic game among consumers. We assume that consumers behave naively  without predictions.}
A natural question is: {Can these simple behaviors (or equivalently, a schedule)} 
eventually lead to a Nash equilibrium {-- a state where no consumer regrets her previous decision? This question is of both theoretical and practical interests. Schedules that lead  to Nash equilibria are called {\em regret-proof}; they guarantee  high consumer satisfaction, which is beneficial to the seller's future marketing.} 

\bigskip\noindent{\bf Results and contribution}\quad
We prove {that it is NP-hard to find} a marketing schedule that maximizes the number of   $\mathcal Y$ (resp. $\mathcal N$) decisions.  Complementary
to the NP-hardness, we  {design $O(n^2)$-time  algorithms for finding schedules that guarantee at least $n/2$ decisions of $\mathcal Y$, and at least
$n/3$ decisions of $\mathcal N$, respectively.}  The numbers $n/2$ and $n/3$   are best possible for any algorithm.  {Let $\alpha$ denote the size of
maximum independent set of $G$. We show that   regret-proof schedules that guarantee  at least $n/2$ decisions of $\mathcal Y$ and at least
$\max\{\sqrt{n+1}-1, (n-\alpha)/2\}$}  decisions of $\mathcal N$, respectively, can be found in time $O(mn^2)$. In contrast,
decentralized consumer choices without a schedule might result in an arbitrarily worse outcome. This can be seen from the star network, where in the
worst case only one consumer chooses the product consistent with the  seller's  objective.

To the best of our knowledge, this paper  is the first attempt to address  the scheduling problems for marketing products with negative externalities (i.e  marketing in rebel social
networks).  Our algorithms for maximizing the number
of $\mathcal Y$ decisions can be extended to deal with the case of promoting one product
where $\mathcal Y$ and $\mathcal  N $ are interpreted as buying and not buying, respectively.

\section{Maximization}\label{sec:maxprofit}
 We study the rebel scheduling problem to maximize seller's profits in Subsections \ref{sec:maxy} and \ref{sec:maxn}, respectively, for the cases of $\mathcal Y$ and $\mathcal N$ having higher net profits.

Throughout we consider      $G=(V,E)$ a connected rebel social network   for which we have $n=O(m)$.  All results can be extended to any network without isolated nodes.  Let   $\pi$ be a schedule for $G$, and $u,v\in V$. We say that $\pi$ schedules  $v\in V$ with decision $\pi[v]\in\{\mathcal Y,\mathcal N\}$, and $\pi$ schedules $u$ before $v$ if $\pi(u)<\pi(v)$.

\subsection{When $\mathcal Y$ is more profitable}\label{sec:maxy}

 It is desirable to find an optimal schedule that maximizes the number of consumers purchasing $\mathcal Y$. Although this turns out to be a very hard task (Theorem \ref{yhard}), we can guarantee that at least half of the consumers choose $\mathcal Y$ (Theorem \ref{lem}). 

\begin{theorem}\label{yhard}
The {rebel scheduling problem for maximizing the number of $\mathcal Y$ decisions} is NP-hard.\end{theorem}
\begin{proof}We prove by reduction from the {\em maximum independent set problem}. Given any instance  of the maximum independent set problem on connected graph {$H=(N,F)$}, by adding some pendant nodes to $H$
we construct in polynomial time a network $G$ (an instance of the rebel scheduling problem): For each node $u\in N$ with degree $d(u)$ in $H$, we add
a set $P_u$ of $d(u)$ nodes, and connect each of them to $u$. The resulting network $G=(V,E)$ is specified by $V:=N\cup (\cup_{u\in N}P_u)$ and
$E:=F\cup (\cup_{u\in N}\{up:p\in P_u\})$, where each node in $V\setminus N=\cup_{u\in N}P_u$ is {\em pendant}, and each node $u\in N$ is {\em
non-pendant} and has exactly $2d(u)$ neighbors: half of them are non-pendant nodes in   $N$ and the other half are the $d(u)$ pendant nodes in  $P_u
$.

We associate every  schedule $ {\pi}$ for $G$ with an integer $\theta( {\pi})$, equal to the number of pendant nodes
which are scheduled (by $ {\pi}$) after their  unique  neighbors. Clearly
\begin{gather}\label{theta}
\theta(\pi)\le|V\setminus N|=2|F| \text{ for any schedule }\pi\text{ of }G.\end{gather}
\begin{clm}\label{clm1}
 For any $u\in N$ and any schedule $\pi$ of $G$, if $\pi$ schedules all nodes in $P_u\cup\{u\}$ with $\mathcal Y$, then (all the $d(u)$ pendant neighbors of $u$  in $P_u$
 have to be scheduled before $u$ with decisions $\mathcal Y$, therefore) all the $d(u)$ non-pendant neighbors of $u$
 have to be scheduled with $\mathcal N$ before $u$ is scheduled.
\end{clm}

Consider $\pi$ being an optimal schedule for $G$. If $\theta(\pi)=0$, then $\pi$ schedules all pendant nodes before their neighbors, and  hence all of these pendant nodes choose $\mathcal Y$.  It follows from Claim \ref{clm1} that $\{ v\in N:\pi[v]=\mathcal Y\}$ is an independent set of $H$. Since $\pi$ is optimal, the independence set is maximum in $H$. Thus, in view of (\ref{theta}), to prove the theorem, it suffices to show the following.

\begin{clm}\label{clm2}
Given an optimal schedule $\pi$ for $G$ with $\theta(\pi)>0$, another   optimal schedule $\pi'$ for $G$ with $\theta(\pi')<\theta(\pi)$ can be found in polynomial time.
\end{clm}
Since $\theta(\pi)>0 $, we can take {$w\in N $} to be the {\em last} non-pendant node   scheduled by $\pi$ earlier than some of its pendant neighbors.
  Under $\pi$, let $P_w'$ ($\emptyset\ne P_w'\subseteq P_w$) be the set  of all pendant neighbors of $w$ that are scheduled after $w$, let $U$ be the set of non-pendant nodes scheduled after $w$, and let $P_U$ be the set of the pendant nodes whose (non-pendant)   neighbors belong to $U$ 
  (possibly $U=\emptyset=P_U$). The choice of $w$ implies that $\pi$ schedules every node in $P_U$   before its   neighbor. Without loss of generality we may assume that under~$\pi$,
\begin{itemize}
\vspace{-2mm}\item (Pendant) nodes   in $P_U$ are scheduled before all other nodes (with  $\mathcal Y$).  
\vspace{-2mm}\item (Pendant) nodes in $P_w'$ are scheduled immediately after $w$ one by one. 
\vspace{-2mm}\item (Non-pendant) nodes in $U$ are scheduled at last.\vspace{-2mm}
\end{itemize}

If $\pi$ schedules $w$ with $\mathcal N$, then at later time it schedules all pendant nodes in $P'_w$ with $\mathcal Y$. 
{Another optimal schedule $\pi'$ (for $G$) with the same outcome as $\pi$ can be constructed as follows: $\pi'$  schedules nodes in $P'_w$ first}, and then schedules other nodes of $V$ in a relative order the same as $\pi$. 
Clearly, $\pi'$ with {$\theta(\pi')\le\theta(\pi)-1$} is the desired schedule.
It remains to consider the case where $\pi$ schedules  $w$ with
\begin{gather}\label{y}
\pi[w]=\mathcal Y.
 \end{gather}
 {It follows that $\pi[p]=\mathcal N$ for all $p\in P_w'$.} Let $\pi'$ be the schedule that first schedules nodes of $V\setminus\{w\}$ in a relative order the same as $\pi$, and  schedules $w$ finally. It is clear that $\theta(\pi')\le\theta(\pi)-1$ {and $\pi'[p]=\mathcal Y$ for all $p\in P_w'$}. We only need to show that $\pi'$ is optimal.

Observe that $\pi'$ first schedules  every $v\in V$ satisfying $\pi(v)<\pi(w)$ with the same decision as in $\pi$ (particularly, all nodes in $P_U$ are scheduled with $\mathcal Y$). {Subsequently, $\pi'$ schedules nodes in $P_w'$ and $ U$ in the same relative order as $\pi$}
Finally $\pi'$ schedules $w$.  Since all pendant nodes in $P'_w $ ($\ne\emptyset$) are scheduled {by $\pi'$ with $\mathcal Y$, and by $\pi$ with $\mathcal N$} 
the  optimality of $\pi'$ would follow if $\pi'$ schedules every node of $U$ with the same decision as $\pi$.

Suppose it were not the case. Let {$u\in U\subseteq N $} be the earliest node in $U$ scheduled by $\pi'$ with a decision $\pi'[u]$ different from $\pi[u]$. It must be the case that 
{$w$ is a non-pendant neighbor of $u$ and $\pi[w]\ne\pi[u]$}. At the time $\pi'$ schedules $u$, all pendant neighbors of $u$ in $P_u\subseteq P_U$ have been scheduled with $\mathcal Y$ and the non-pendant neighbor $w$ has not been scheduled,  it follows from Claim \ref{clm1} that $\pi'[u]=\mathcal N$. {As $\pi[u]\ne\pi'[u]$ and $\pi[w]\ne\pi[u]$, we have $\pi[w]=\pi'[u]=\mathcal N$,} 
 a contradiction to (\ref{y}). The optimality of {$\pi'$} is established, which proves
Claim~\ref{clm2} and therefore Theorem~\ref{yhard}.
\end{proof}

We next design an algorithm for finding a schedule that ensures  at least $n/2$ decisions of $\mathcal Y$. 
The algorithm iteratively constructs
 a node set $A$ for which there exist two schedules $\pi'$ and $\pi''$  scheduling each node in $A$ with different decisions.
  In the end, at least  half nodes of $A$ can be scheduled by either $\pi'$ or $\pi''$ with $\mathcal Y$ decisions.
 Subsequently,  the nodes outside $A$, which form an independent set, will all choose $\mathcal Y$ (in an arbitrary order).

\begin{algorithm}\label{al1}
{\em Input}: 
{Network} $G=(V,E)$. {\em Output}: {Partial schedule $\pi$ for $G$.}
\end{algorithm}
\vspace{-0mm}\hrule\begin{enumerate}
\vspace{-0mm}\item Initial setting: $A\leftarrow\emptyset, t\leftarrow 1$, $\pi'\leftarrow$ a null schedule

 \vspace{-2mm} \item \textbf{While} $\exists$ $w\in V\setminus A$ which has different numbers of neighbors  in $A$ choosing $\mathcal Y$ and $\mathcal N$ respectively under $\pi'$ {\bf do}
  \label{findnode}

 \vspace{-2mm}\item  \hspace{3mm}  {\sc schedule $w$:}  $\pi'(w)\leftarrow t$, $\pi''(w)\leftarrow t$;\label{schedulenode}

  \vspace{-1mm} \hspace{3mm}  $A\leftarrow A\cup \{w\}$,  $t\leftarrow t+1$

\vspace{-2mm}\item    \textbf{End-while}\label{endnode}
\vspace{-2mm}  \item {\bf If}  $\exists$ $uv\in E$ with $u,v\notin A$  \label{findedge}

 \hspace{1mm}  {\bf then}
{\sc schedule $uv$:} $ \pi'(u)\leftarrow t,\pi'(v)\leftarrow t+1,\pi''(v)\leftarrow t,\pi''(u)\leftarrow t+1$;\label{scheduleedge}

 \hspace{10mm}  $A\leftarrow A\cup \{u,v\}$,  $ t\leftarrow t+2$;

  \hspace{10mm} Go back to Step \ref{findnode}.

  \item Let $\pi$ be $\pi'$ or $\pi''$ whichever schedules {more nodes with $\mathcal Y$ (break tie arbitrarily)} 
      \label{takebetter}


\vspace{2mm}
\hrule
\end{enumerate}

For convenience, we reserve symbol ``{\sc schedule}'' for the scheduling {(constructing $\pi$ and $\pi''$)} as conducted at Steps \ref{schedulenode} and \ref{scheduleedge} in Algorithm \ref{al1}. Similarly, we also say    ``{\sc schedule} a node'' and  ``{\sc schedule} an edge'' with the implicit understanding that the node and the edge satisfy the conditions in Step \ref{findnode} and Step \ref{findedge} of Algorithm~\ref{al1}.
\begin{clm}\label{clm3}
 $\pi'[v]=\mathcal Y$ if and only if
$\pi''[v]=\mathcal N$ for all $v\in A$.
\end{clm}
\begin{proof}
 The algorithm enlarges $A$ gradually at Steps \ref{schedulenode} and \ref{scheduleedge}, producing a sequence of node sets $A_0=\emptyset$, $A_1$, \ldots, $A_{\ell}=A$. We prove by induction on $k$ that  $\pi'(v)=\mathcal Y$ if and only if
$\pi''(v)=\mathcal N$ for all $v\in A_k$, $k=0,1,\ldots,\ell$. The base case of $k=0$ is trivial.

Suppose that $k\ge 1$ and the statement is true for $A_{k-1}$. In case of $A_k$ being produced at Step \ref{findnode}, suppose $w$ has $n_1$ (resp.
$n_2$) neighbors in $A_{k-1}$ choosing $\mathcal Y$ (resp. $\mathcal N$) under $\pi'$. By hypothesis, $w$ has $n_1$ (resp. $n_2$) neighbors in
$A_{k-1}$ choosing $\mathcal N$ (resp. $\mathcal Y$) under $\pi''$. Since $n_1\ne n_2$, we see that $\pi'[w]=\mathcal Y$ if and only if
$\pi''[w]=\mathcal N$. In case of $A_k$ being produced at Step \ref{scheduleedge}, both $u$ and $v$ have equal number of neighbors in $A_{k-1}$
choosing $\mathcal Y$ and $\mathcal N$, respectively, under $\pi'$, due to the implementation of the while-loop at Steps \ref{findnode}--\ref{endnode}.
By hypothesis both $u$ and $v$ have equal number of neighbors in $A_{k-1}$ choosing $\mathcal Y$ and $\mathcal N$, respectively, under $\pi''$. It
follows from $uv\in E$ that $\pi'[u]=\pi''[v]=\mathcal Y$ and $\pi'[v]=\pi''[u]=\mathcal N$. In either case, the statement is true for $A_k$, proving
the claim.
\end{proof}

\begin{clm}\label{ind} \begin{mylabel}\item[(i)] At least half nodes of $A$ are scheduled by $\pi$ with $\mathcal Y$  (by Step \ref{takebetter}).
\item[(ii)]The nodes in $V\setminus A$ (if any) form an independent set of $G$ (by Step \ref{findedge}).
\item[(iii)] Each node in $V\setminus A$ has an equal number of neighbors in $A$ choosing $Y $ and $\mathcal N$, respectively, under $\pi'$ (by Steps
 \ref{findnode}--\ref{endnode}), and under $\pi''$ (by Claim \ref{clm3}), and hence  under  $\pi$  (by Step \ref{takebetter}).\qed
\end{mylabel}
\end{clm}

\begin{theorem}\label{lem}
    {A schedule    that ensures at least $n/2$ decisions of $\mathcal Y$ can be found in $O(n^2)$ time.}
\end{theorem}
\begin{proof} {It follows from Claim \ref{ind}(ii) and (iii) that $\pi$ can be  {extended to a} schedule for $G$ such that all node in $V\setminus A$ choose $\mathcal Y$. By Claim \ref{ind}(i), the outcome   contains at least $n/2$ decisions of $\mathcal Y$.}

Next we show the time complexity. Algorithm \ref{al1} keeps an $n\times 2$ array $[\delta_v,\delta_v']$, $v\in V$, where $\delta_v$ represents the difference between the numbers of neighbors of node $v$ in $A$ choosing $\mathcal Y$ and $\mathcal N$, and $\delta_v'$ represents the number of neighbors of node $v$ in $V\setminus A$. The initial setting of the array $[\delta_v,\delta_v']=[0,$ the degree of $v$ in $G]$, $v\in V$, takes   $O(n^2)$ time.  Step 1 is to find a node $w\notin A$ with $\delta_w\ne0$ by visiting $\delta_v$, $v\in V\setminus A$. Step \ref{findedge} is to find a node $u\in V\setminus A$ with $\delta'_u\ge1$ and then find a neighbor $u\in V\setminus A$ of $v$. The search in both  Steps \ref{findnode} and \ref{findedge} takes   $O(n)$ time. Each time Algorithm \ref{al1} adds a node $v$ to $A$, the algorithm updates the entries  of $v$'s   neighbors   in the array, which   takes   $O(n)$ time. Since we can add at most $n$ nodes to $A$, Algorithm \ref{al1}  terminates in $O(n^2)$ time.
\end{proof}

 The tightness of $n/2$ in the above theorem  can be seen from the case where the network $G$ is a complete graph. Moreover, the theorem implies that Algorithm~\ref{al1} is a 2-approximation algorithm  for the rebel scheduling problem for maximizing $\mathcal Y$ decisions.

 \begin{remark}\label{rm}
 It is worth noting that Algorithm~\ref{al1} can be used to solve the scheduling problem when only one product is promoted, where a consumer buys the product only if  at least a half of her neighbors do not have the product. Given a schedule $\pi$ output by Algorithm \ref{al1}, $\pi$ specifies an order of consumers who choose $\mathcal Y$. All these consumers will buy the product if the seller promotes the product to them according to this order.
\end{remark}

\subsection{When $\mathcal N$ is more profitable}\label{sec:maxn}
 In this subsection, the marketing scheduling is to maximize the number of $\mathcal N $ decisions.  
 By reduction from the {\em bounded occurrence MAX-2SAT} problem  (see Appendix \ref{ap:a}), 
 we obtain the following   NP-hardness.
\begin{theorem}\label{maxN}
The  rebel  scheduling problem for maximizing the number of  $\mathcal N$ decisions is NP-hard.\qed
\end{theorem}
Next, we  {design a  $3$-approximation algorithm for finding in $O(n^2)$ time a schedule which ensures at least $n/3$ decisions of $\mathcal N$.}
This is accomplished by  a {refinement} 
 of Algorithm \ref{al1} with some preprocessing.

 {The following terminologies will be used in our discussion. Given a graph $H$ with node set $U$, let $R,S\subseteq U$ be two node subsets. We say that $R$ {\em dominates}  $S$ if every node in $S$ has at least a neighbor in $R$. We use $H\setminus R$ to denote the graph obtained from $H$ by deleting all nodes in $R$ (as well as their incident {links}). Thus $H\setminus R$ is the subgraph of $H$ induced by $U\setminus R$, which we also denote as $H[U\setminus R]$.}

\paragraph{Preprocessing.}
{Given a connected social network $G=(V,E)$,} 
let $X$ be any maximal independent set {of $G$.  It is clear that }

\begin{itemize}
\vspace{-2mm}\item {$X$ and $Y:=V\setminus X$ are disjoint node sets dominating each other.}\vspace{-2mm}
\end{itemize}
We will partition {$X$ into $X_1,\ldots,X_{\ell}$ and $Y$ into $Y_0,Y_1,\ldots,Y_{\ell}$} 
for some positive integer $\ell$ such that  Algorithm \ref{al1} schedules $X_i\cup Y_i$ before $X_{i-1}\cup Y_{i-1}$ for all $i=\ell,\ell-1,\ldots,2$.
\begin{itemize}
\vspace{-2mm}\item Set  {$G_0=G$ and}
$X_0=\emptyset$. {Find} $Y_0\subseteq Y$ such that  $Y\setminus Y_0$ is a {\em minimal} set 
 that dominates $X\setminus X_0$ {($=X$)  in graph $G_0$}.
 \item Set graph $G_1:=G \setminus (X_0\cup Y_0)=G[(X\setminus X_0)\cup(Y\setminus Y_0)]$.\vspace{-2mm}
 \end{itemize}
 The minimality of  $Y\setminus Y_0$  implies that in 
 {graph} $ G_1$ every node in $Y\setminus Y_0$  is {adjacent  to  at least one pendant node} 
 in $X\setminus X_0$.
 \begin{itemize}
\vspace{-2mm}\item  Let $X_1\subseteq X\setminus X_0$ consist  of all {pendant} nodes of $ G_1$  
contained in {$X\setminus X_0$}.\vspace{-2mm}
 \end{itemize}
 If $X\setminus (X_0\cup X_1)\ne\emptyset$, {then} $Y\setminus Y_0$ still dominates $X\setminus (X_0\cup X_1)$, {and} we  repeat the above process {with $G_1$, $X\setminus X_0$, $Y\setminus Y_0$ in place of $G_0$, $X$, $Y$, respectively,} and produce $Y_1$, $G_2$, $X_2$ in place of $Y_0$,  $G_1$, $X_1$.

 Inductively, for $i=1,2,\ldots$, {given graph $G_i=G\setminus\cup_{j=0}^{i-1}(X_j\cup Y_j)=G[(X\setminus\cup_{j={0}}^{i-1} X_j)\cup(Y\setminus\cup_{j={0}}^{i-1} Y_j)]$, where {$Y\setminus\cup_{j=0}^{i-1} Y_j$} is a minimal set  dominating $X\setminus\cup_{j=1}^{i-1} X_j$, and $X_i$ the set of all pendant nodes of $G_i$ contained in $X\setminus\cup_{j={0}}^{i-1} X_j$,  when $X\setminus\cup_{j={0}}^{i} X_j\ne\emptyset$, we can}


  \begin{itemize}
\vspace{-2mm}\item
 {Find} $Y_i\subseteq (Y\setminus \cup_{j=0}^{i-1}Y_j)$ such that $Y\setminus\cup_{j=0}^{i}Y_j$ is a {\em minimal} set 
 that dominates $X\setminus \cup_{j=0}^{i}X_j$ {in graph $G_i$}.

\item
{Set graph $G_{i+1}:=G\setminus\cup_{j=1}^i(X_j\cup Y_j)=G[(X\setminus\cup_{j=0}^{ i}X_j)\cup(Y\setminus\cup_{j=0}^{ i}Y_j)]$.}

\item
Let $X_{i+1}\subseteq X\setminus\cup_{j=0}^{i}X_j$ consist of all {pendant} 
nodes of $ G_{i+1}$  that are contained in {$X\setminus\cup_{j=0}^{i}X_j$}. \vspace{-3mm}
 \end{itemize}
{The procedure terminates at $i=\ell$ for which we have $X\setminus\cup_{j=0}^\ell X_j=\emptyset$, and} 
\begin{center}
 $G_i=G[(\cup_{j=i}^{\ell}Y_j)\bigcup(\cup_{j=i}^{\ell}X_j)]$ for $i=0,1,\ldots,\ell$; in particular $G_0=G$.
  \end{center}
 Note that $G_i\subseteq G_{i-1}$ for $i=\ell,\ell-1,\ldots,1$,
 $Y\setminus Y_0$ is the disjoint union of $Y_1,\ldots, Y_{\ell}$, and $X$ is  the disjoint union of $X_1,\ldots, X_{\ell}$.
 The minimality of $\cup_{j=i}^{\ell}Y_j=Y\setminus\cup_{j=0}^{i-1}Y_j$ implies that in graph $ G_i$ every node in $\cup_{j=i}^{\ell}Y_j$ is {adjacent to} at least one pendant  node   in $X_i$.

\begin{clm}\label{clm7}
 For any   $i=\ell,\ell-1,\ldots,1$, in the subgraph $G_i$, all nodes  in $X_{i}$ are pendant, and every node in $Y_i$ is adjacent to at least one node in $X_i$.
\end{clm}


\paragraph{Refinement.} Next we show that Algorithm \ref{al1}  can be implemented {in a way that} 
all nodes of  subgraph $G_1$  are scheduled. If {the implementation has led to at least $n/3$ decisions of $\mathcal N$,} 
we are done; otherwise, {due to the maximality of the independent set $X$}, we can easily find another schedule that  makes at least $n/3$ nodes choose $\mathcal N$.

\begin{algorithm}\label{al2}
{\em Input}: {Network  $G=(V,E)$ together
with $G_j,X_j,Y_j$, $j=0,1,\ldots,\ell$.} 
{\em Output}: {Partial schedule $\pi$ for $G$.} 
\end{algorithm}
\vspace{-2mm}\hrule
\begin{enumerate}
\vspace{-2mm}\item {Initial setting: $A\leftarrow\emptyset$}

\vspace{-2mm}\item {\bf For} $i=\ell$ {\bf downto} 0 {\bf do}

\vspace{-2mm}  \item  \hspace{1mm}  \textbf{While} in the subgraph $G_i$, $\exists$ $w\in (X_i\cup Y_i)\setminus A$ which has different numbers

\vspace{-1mm}\hspace{1mm}  of neighbors  in $A$ choosing $\mathcal Y$ and $\mathcal N$ respectively \textbf{do} \label{stp2}

\vspace{-2mm}\item \hspace{4mm} {\sc schedule $w$}; $A\leftarrow A\cup\{w\}$\label{schedulew}

\vspace{-2mm}\item \hspace{1mm} {\bf End-while}

 \vspace{-2mm} \item \hspace{1mm} {\bf If} {$\exists$ edge $uv$ of $G_i$ with $u,v\notin A$} 

\vspace{-2mm}\item  \hspace{4mm} {\bf then} {\sc schedule $uv$};  $A\leftarrow A\cup \{u,v\}$;     Go back to Step \ref{stp2}.\label{stp3}

\vspace{-2mm} \item {\bf End-for}

\vspace{-2mm} \item Let $\pi$ be a schedule for $G[A]$ that schedules at least $\frac12|A|$ nodes with $\mathcal N$\label{output}
  \vspace{2mm}
  \hrule
   \end{enumerate}

The validity of Step \ref{output} is guaranteed by Claim \ref{clm3}. Since $X\cap Y_0=\emptyset$, the following claim implies $X\subseteq A$.

\begin{clm}\label{contain}
 $V\setminus A\subseteq Y_0$. 
\end{clm}
\begin{proof}
We only need to show  {that each node $w\in X_k\cup Y_k$ $(k=1, 2,\cdots, \ell)$ is selected to $A$ when $i=k$ in Algorithm \ref{al2}.}

 In case of  $w\in X_k$, it is pendant and has only one neighbor $u$ in subgraph $G_k$. If $u\in A$ when $w$ is checked at Step \ref{stp2}, then $w$ is selected to $A$ at Step \ref{schedulew}; Otherwise, $w$ and $u$ will be selected to $A$ at the same time in Step \ref{stp3}.

 {In case of} 
 $w\in Y_k$, {by Claim \ref{clm7},  $w$ is adjacent to a  pendant node $v\in X_k$ of $G_k$.} If, {when   checked {at Step \ref{stp2}},}  $w$ has different numbers of neighbors in $A$ choosing $\mathcal Y$ and $\mathcal N$, then it is selected to $A$ {at Step \ref{schedulew}}; 
 otherwise,  node $v$ must have not been selected to $A$, {and subsequently $w$ and $v$ are put into $A$ together at Step~\ref{stp3}.}
\end{proof}
If $|A|> 2n/3$, then,  {by extending partial schedule $\pi$ output by Algorithm \ref{al2}, we  obtain} 
a schedule which makes at least $n/3$ nodes 
choose $\mathcal N$. Otherwise, {$|V\setminus A|\ge n/3$, and all nodes in $V\setminus A$ can be scheduled with $\mathcal N$ as follows: Schedule firstly the nodes in the maximal independent set $X$ {(all of them choose $\mathcal Y$)}; secondly the nodes in $V\setminus A$, and}  
finally all the other nodes. Recall that $X$ dominates every node in {$Y\supseteq Y_0$. It follows from Claim \ref{contain} that} $X$ dominates
$V\setminus A$. {As $V\setminus A$ is an independent set in $G$ (by Claim \ref{ind}(ii)), the decisions of all nodes in $V\setminus A$ are $\mathcal
N$.}  We show in Appendix \ref{ap:b} that Algorithm \ref{al2} runs in square time, which implies the following.

\begin{theorem}
{A schedule    that ensures at least $n/3$ decisions of $\mathcal N$ can be found in $O(n^2)$ time.}\label{th4}\qed
\end{theorem}

{The tightness of $n/3$ 
can be seen from a number of disjoint triangles linked by a path, where each triangle has exactly two nodes of degree two.}
\section{Regret-proof schedules}\label{sec:regretfree}
 We are to find  regret-proof schedules, where every consumer, given the choices of other consumers in the outcome of the schedule, would prefer the product she bought to the other. Using  link cuts as a tool, we
design algorithms for finding   regret-proof schedules that ensure  at least $n/2$ decisions of $\mathcal Y$ and at least $\sqrt{n+1}-1$ decisions of $\mathcal N$, respectively. 

\subsection{Stable cuts}
 Given $G=(V,E)$, let $R$ and
$S$ be two disjoint subsets of $V$. We use $[R,S]$ to denote the set of links (in $E$) with one end in $R$ and the other in $S$. If $R\cup S=V$, we
call $[R,S]$  a  {\em link  cut} or simply a {\em cut}.  For a node $v\in V$, we use $d_S(v)$ to denote the number of neighbors of $v$
contained in $S$. {\em Each schedule $\pi$ for $G$ is  associated with a cut $[S_1,S_2]$ of $G$} defined by its outcome: $S_1$ (resp. $S_2$) is the
set of consumers scheduled with $\mathcal Y$ (resp. $\mathcal N$).
A schedule $\pi$ is {\em regret-proof} if and only if its associated cut $[S_1,S_2]$ is
{\em stable}, i.e., satisfies the following conditions:
\begin{equation} d_{S_2}(v)\ge d_{S_1}(v)\text{ for any }v\in S_1, \text{ and }d_{S_1}(v)>d_{S_2}(v)\text{ for any }v\in S_2.\label{stable}\end{equation}
Note that $S_1$ and $S_2$ are asymmetric. For clarity, we call $S_1$ the {\em leading set} of cut $[S_1,S_2]$. Any node that violates (\ref{stable}) is called {\em violating} (w.r.t. $[S_1,S_2]$). 

A basic operation in our algorithms is ``enlarging'' unstable cuts by moving {``violating''} nodes from one side to the other. Let $[S_1,S_2]$ be an unstable cut of $G$ for which some $v\in S_i$ ($i=1$ or 2) is violating. We define {\em type}-$i$ {\em move} of $v$ (from $S_i$ to $S_{3-i}$) to be the setting:  $S_i\leftarrow S_i\setminus\{v\}$,  $S_{3-i}\leftarrow S_{3-i}\cup\{v\}$, which changes the cut. The violation of (\ref{stable}) implies
\begin{mylabel2}
\vspace{-1mm}\item[(M1)] {type}-$1$ {move} increases the cut size, and downsizes the leading set;
\item[(M2)] {type}-$2$ {move} does not decrease the cut size, and enlarges the leading set.\vspace{-1mm}
\end{mylabel2}

Both types of moves are collectively called {\em moves}.
{Note that   moves  are only defined for violating nodes, and the cut size $|[S_1,S_2]|$ is nondecreasing under moves.} To find a stable cut, our
algorithms work with a cut $[S_1,S_2]$ of $G$ and change it by moves sequentially.  By (M1) and (M2), the number $m_1$ of {type}-$1$ moves is $O(m)$. {Moreover, we have the following observation.}

\begin{lemma}\label{movement}\begin{mylabel}\item[(i)]
From any given cut {of size $s$}, $O(m_1+n)$ moves produce a stable cut (i.e., a cut without violating nodes) of  {of size at least $s+m_1$}.
\item[(ii)]If the leading set of the stable cut produced is smaller than that of the given cut, then the number of type-2 moves is smaller than that of type-1 moves.\qed
\end{mylabel}
\end{lemma}

 As a byproduct of (M1) and (M2), one can easily  deduce  that the rebel game on a network,  where each rebel
 switches  between two choices in favor of the minority choice of her neighbors, is a potential game and thus possesses a Nash
equilibrium. The potential function is defined as the size of the cut between the rebels holding different choices.


The following data structure is employed
for efficiently identifying  violations as well as verifying the stability of the cut. For given
cut $[S_1,S_2]$,
we create
in $O(m)$ time a $2$-dimensional array $(i(v), \Delta(v))$,  $v\in V$, of length $n$, where
$i(v)\in\{1,2\}$ is the set index satisfying $S_{i(v)}\ni v$,
and $\Delta(v)=d_{S_{3-i(v)}}(v)-d_{S_{i(v)}}(v)$ together with $i(v)$ is the indicator of
whether $v$ is violating. A node $v$ is violating if and only if $\Delta(v)<0$ when $i(v)=1$ or
$\Delta(v)\le 0$ when $i(v)=2$. Therefore, in $O(n)$ time we can find a violating node $v$ (if
any) and move it. After the move, we update the array (to be the one for the current cut) in $O(n)$ time by modifying the entries corresponding to $v$ and its neighbors. Without consideration of the $O(m)$ time creation of the array, we have the following lemma.

\begin{lemma}\label{timen}
In $O(n)$ time, either the current cut is verified to be stable, or a move is found and \mbox{conducted.\!\!\!\qed}
\end{lemma}

The following procedure, as a subroutine of our algorithm, finds a  stable cut {whose leading set contains at least half nodes of $G$}.

\begin{procedure}\label{proc:2}
{\em Input}: Cut $[S_1,S_2]$ of $G$. {\em Output}:
Stable cut $[T_1,T_2]:=\text{\sc Prc\ref{proc:2}}(S_1,S_2)$
\end{procedure}
\vspace{0mm}\hrule
\begin{enumerate}
\vspace{-1mm}\item \textbf{Repeat}
\vspace{-2mm}\item\hspace{1mm} {\bf If} $|S_1|<n/2$ {\bf then}   $[S_1,S_2]\leftarrow [S_2,S_1]$  \hfill{\small // swap $S_1$ and $S_2$}\label{swap}

\vspace{-2mm}
\item\hspace{1mm} \textbf{While} $\exists$ violating node  $v  $ w.r.t. $[S_1,S_2]$   \textbf{do} move $v$\label{findviolation} \hfill{\small // $[S_1,S_2]$ is changing}\label{while}

\vspace{-2mm}
\item \textbf{Until} $|S_1|\ge n/2$ \label{until}

\vspace{-2mm}\item  Return $[T_1,T_2]\leftarrow [S_1,S_2]$

\vspace{1mm}
  \hrule
   \end{enumerate}

\begin{lemma}\label{lem2}
{Procedure  \ref{proc:2}  produces in $O(tn+n^2)$  time  a stable cut $[T_1,T_2]$ of $G$ such that $|T_1|\ge n/2$, where $t=|[T_1,T_2]| - |[S_1,S_2]|\ge0$.}
\end{lemma}
\begin{proof}  {It follows from  Lemma \ref{movement}(i) that}  there are a   number $m_1'$ ($\le m$) of type-1 moves in total, and {$|[T_1,T_2]|\ge|[S_1,S_2]|+m'_1$. By Lemma \ref{timen}, it
suffices to show that there are a total of $O(m_1'+n)$   moves}.

 Observe from Step \ref{swap} that each (implementation) of the while-loop {at Step \ref{while}} starts with a cut whose leading set has at least $n/2$ nodes. If this while-loop ends with a smaller leading set, by Lemma \ref{movement}(ii) it must be the case that the while-loop conducts type-1 moves more times than {conducting} type-2 moves. Therefore after $O(m_1')$ moves, the procedure {either terminates, or} implements a while-loop that ends with a leading set $S_1$ not smaller than one at the beginning of the while-loop. In the latter case, the until-condition at Step~\ref{until} is satisfied, and the procedure terminates.  The number of moves conducted by the last while-loop is $O(m_1'+n)$ as  implied by  Lemma~\ref{movement}(i).
\end{proof}

\subsection{$\mathcal Y$-preferred   schedules}
 When $\mathcal Y$ is more profitable,  the basic idea behind our algorithms for finding regret-proof schedules goes as follows: Given a stable cut $[S_1,S_2]$, we try to schedule nodes in $S_1$ with $\mathcal Y$  and nodes in
$S_2$   with $\mathcal N$ whenever possible. If not all nodes can be scheduled {this way},   we obtain another stable cut of larger size, from which
we repeat the process. In the following pseudo-code description, scheduling an {\em unscheduled} node changes the node to be {\em scheduled}.

\begin{algorithm}\label{al4}
{\em Input}:  {Cut $[R_1,R_2]$ of network $G$. {\em Output}: A schedule for $G$}
\end{algorithm}
\vspace{-0mm}\hrule
\begin{enumerate}
\item \vspace{-1mm}  Initial setting:  $\mathcal D_1\leftarrow \mathcal Y $, $\mathcal D_2\leftarrow \mathcal N $; $T_i\leftarrow\emptyset$, $S'_i\leftarrow R_i\setminus T_i$  ($i=1,2$)

\vspace{-2mm}\item {\bf Repeat}

  \vspace{-2mm}\item\hspace{1mm} {$[S_1,S_2]\leftarrow \text{\sc Prc\ref{proc:2}}(S'_1\cup T_2,  S_2'\cup T_1)$ \hfill{\scriptsize // $|[S_1,S_2]|\ge|[S'_1\cup T_2,  S_2'\cup T_1]|$}}\label{s1s2}

  \vspace{-2mm}\item \hspace{1mm} Set all nodes  of $G$  to be unscheduled

 \vspace{-2mm} \item\hspace{1mm}
  \textbf{While} $\exists$    unscheduled $v\!\in\! S_i$ ($i\!\in\!\{1,2\}$) whose decision is  $\mathcal D_i$  \textbf{do} schedule $v$\label{sv}

 \vspace{-2mm} \item\hspace{1mm}     {$T_i\leftarrow \{$scheduled nodes with decision $\mathcal D_i\}$, $S'_i\leftarrow S_i\setminus T_i$ ($i=1,2$)}
       \hfill{\scriptsize //$T_i\subseteq S_i$}\label{ti}

 \vspace{-2mm} \item {\bf Until} $S'_1=\emptyset$  \hfill{\scriptsize // {Until all nodes of $G$ are scheduled}}\label{untils}

 \vspace{-2mm} \item  Output the final schedule for $G$
 \vspace{1mm}
  \hrule
   \end{enumerate}
{Note that cuts $[S_1,S_2]$  returned by Procedure \ref{proc:2} at Step \ref{s1s2} are stable.} At the end of Step \ref{ti}, if $S'_1=\emptyset$, then $S'_2=\emptyset$  {(otherwise, every node    $v\in S'_2\subseteq S_2$ satisfies $d_{S_1}(v)=d_{T_1}(v)\le d_{T_2}(v)\le d_{S_2}(v)$, saying that    $[S_1,S_2]$ is not  stable.)} Thus the condition in Step \ref{untils} is equivalent to saying ``until all nodes of $G$ are scheduled''.
\begin{theorem}\label{the}
Algorithm \ref{al4} finds in $O(mn^2)$ time a regret-proof schedule with at least $n/2$ decisions of $\mathcal Y$.
\end{theorem}
\begin{proof} Consider Step \ref{ti}
 setting $ S'_1\ne\emptyset$. Since nodes in $S_1'\cup S'_2$ cannot be scheduled,  we have  $d_{T_1}(v)>d_{T_2}(v)$ for every $v\in S'_1=S_1\setminus T_1$ and $d_{T_2}(v)\ge d_{T_1}(v) $ for any $v\in S'_2=S_2\setminus T_2$, {which gives}
 \begin{eqnarray*}
 0&<& \sum_{v\in S'_1} (d_{T_1}(v)-d_{T_{2}}(v))+ \sum_{v\in S'_2} (d_{T_2}(v)-d_{T_{1}}(v))\\
&=&(|[S'_1,T_1]|-|[S'_1,T_2]|)+(|[S'_2,T_2]|-|[S'_2,T_1]|)\\
&=&|[S'_1\cup T_2,S'_2\cup T_1]|-|[S_1'\cup T_1,S_2'\cup T_2]|.
\end{eqnarray*}
Thus cut $[S'_1\cup T_2,S'_2\cup T_1]$ {has its size  $t>|[S_1'\cup T_1,S_2'\cup T_2]|=|[S_1,S_2]|$.} Subsequently, at Step \ref{s1s2}, with  input $[S'_1\cup T_2,S'_2\cup T_1]$, Procedure~\ref{proc:2} returns a new cut $[S_1,S_2]$, of size at least $t$, {which is larger than   the old  one}. It follows that the repeat-loop can only repeat a number $k $ ($\le m$) of times. 

 {By Lemma \ref{lem2}, for $i=1,2,\ldots,k$, we assume that Procedure \ref{proc:2} in the $i$-th repetition (of Steps \ref{s1s2}--\ref{ti}) returns  in $O(t_in+n^2)$  time a cut whose size is $t_i$ larger than the size of its input. Then   $\sum_{i=1}^kt_i\le m$, and overall
   Step \ref{s1s2}   takes  $O(\sum_{i=1}^k(t_in+n^2))=O(mn^2)$ time}. The overall running time follows from the fact that $O(n^2)$ time is enough for finishing  a whole while-loop at Step  \ref{sv}.

Note {from Lemma \ref{lem2}} that  the final cut $[S_1,S_2]$ produced by Procedure \ref{proc:2} is stable and satisfies  $|S_1|\ge n/2$. {Since $[S_1,S_2]$ is the cut associated with the final schedule output, the theorem is proved.}
\end{proof}

 Similar to Remark \ref{rm}, the output of Algorithm~\ref{al4} specifies a  regret-proof  schedule for marketing one product such that at least a half of consumers buy the product.

\subsection{$\mathcal N$-preferred schedules}
 The goal of this subsection is to  design an algorithm for  finding a regret-proof schedule  with  as many $\mathcal N$ decisions as possible.  In the following Algorithm \ref{al5}, we work on a dynamically changing cut $[S_1,S_2]$ of $G$ whose size keeps nondecreasing.  Our algorithm consists of 2-layer nested repeat-loops.
  \begin{itemize}
  \item {\em Inner loop}: From any $[S_1,S_2]$, by moving violating nodes, we make it stable.
  Then we try to schedule nodes in $S_1$ with $\mathcal Y$ and nodes in $S_2$ with $\mathcal N$ whenever possible. If not all nodes can be scheduled, we reset $[S_1,S_2]$ to be a larger cut, and repeat; otherwise we obtain a schedule with associated cut $[S_1,S_2]$.

  \item {\em Outer loop}: After obtaining a schedule, we swap $S_1$ and $S_2$, and repeat.
  \item {\em Termination}: We stop when we obtain (consecutively) two schedules whose associated cuts have equal size.
  \item {\em Output}: Among the obtained schedules, we output the best one with a maximum number of $\mathcal N$ decisions
  \end{itemize}
  In the following pseudo-code, we use $r$ and $s$ to denote the sizes of cuts associated with the two schedules we find consecutively. We use
  $K$ to denote the largest number of $\mathcal N$ decisions we currently achieve  by some schedule.

\begin{algorithm}\label{al5}
{\em Input}: Network $G$. \quad {\em Output}: A regret-proof schedule for $G$.
\end{algorithm}
\hrule
\begin{enumerate}
\vspace{-2mm}\item   $\mathcal D_1\leftarrow \mathcal Y$, $\mathcal D_2\leftarrow \mathcal N $;  $[S_1,S_2]\leftarrow$ any  cut of $G$; $s\leftarrow 0$; $K\leftarrow0$

\vspace{-2mm}\item {\bf Repeat}

\vspace{-2mm}\item \hspace{1mm} $r\leftarrow s$;  $[S_1,S_2]\leftarrow[S_2,S_1]$  \hfill{\scriptsize // swap $S_1$ and $S_2$}\label{swp}
\vspace{-2mm}\item \hspace{1mm} {\bf Repeat}\label{inner}

\vspace{-2mm}  \item\hspace{3mm} \textbf{While} $\exists$ violating node $v$ w.r.t.  $[S_1,S_2]$ \textbf{do} move $v$ \hfill{\scriptsize //make $[S_1,S_2]$ stable}\label{nonstable}

\vspace{-2mm}  \item \hspace{3mm} Set all nodes  of $G$  to be unscheduled \label{uns}

\vspace{-2mm}  \item\hspace{3mm}
  \textbf{While} $\exists$    unscheduled $v\!\in\! S_i$ ($i\!\in\!\{1,2\}$) whose decision is  $\mathcal D_i$  \textbf{do} schedule~$v$\label{schedulev}

\vspace{-2mm}  \item\hspace{3mm}     {$T_i\leftarrow \{$scheduled nodes with decision $\mathcal D_i\}$, $S'_i\leftarrow S_i\setminus T_i$ ($i=1,2$)}\label{tis}
       \hspace{1mm}{\scriptsize //$T_i\subseteq S_i$}

 \vspace{-2mm}      \item \hspace{3mm} \textbf{If} $S_1'\ne\emptyset$ \textbf{then} $[S_1,S_2]\leftarrow[S'_1\cup T_2,S_2'\cup T_1]$  \hfill{\scriptsize //reset $[S_1,S_2]$ to be a larger cut}\label{reset}

 \vspace{-2mm} \item \hspace{1mm} {\bf Until} $S'_1=\emptyset$  \hfill{\scriptsize // until all nodes are scheduled}\label{untilempty}
 \vspace{-2mm}  \item \hspace{1mm} \textbf{If}  $|S_2|>K$ \textbf{then} $\pi\leftarrow $ the current schedule, $K\leftarrow|S_2|$\label{record}

 \vspace{-2mm}\item \hspace{1mm}   $s\leftarrow |[S_1,S_2]|$

 \vspace{-2mm}\item  \textbf{Until} $r=s$  \hfill{\scriptsize //until Steps \ref{inner}--\ref{untilempty} produce 2 schedules whose associated cut have equal size}\label{terminate}

\vspace{-2mm}\item Output $\pi$

\vspace{1mm}

\hrule
\end{enumerate}


Throughout the algorithm, the size of $[S_1,S_2]$ keeps nondecreasing, and may increase at Step~\ref{nonstable} (see Lemma \ref{movement}) and Step \ref{reset}.  Note that Steps \ref{schedulev} and \ref{tis} are exactly Steps \ref{sv} and \ref{ti} of Algorithm \ref{al4}.  So, as shown in the proof of Theorem \ref{the}, the resetting of $[S_1,S_2]$ at Step \ref{reset} increases the cut size.

\begin{lemma}\label{run}
Algorithm \ref{al5} runs in $O(mn^2)$ time.
\end{lemma}
\begin{proof}  An implementation of the inner repeat-loop executes Steps \ref{nonstable}-\ref{reset} at most $m$ times. From the termination condition at Step \ref{terminate}, we see that the outer repeat-loop runs $O(m)$ times. Furthermore, we may assume that the algorithm implements Steps \ref{nonstable}-\ref{reset} for a number of ${\ell}$ times, where the $i$-th implementation of Step \ref{nonstable} (resp. Step \ref{reset}) increases the cut size by $m_i$ (resp.  $n_i$), $i=1,2,\ldots,{\ell}$, such that $m_i+n_i\ge1$ for $i=1,2,\ldots,{\ell}-1$ and $m_{\ell}+n_{\ell}=0$. Since $\sum_{i=1}^{\ell}(m_i+n_i)\le m$, we have ${\ell}=O(m)$ and $\sum_{i=1}^{\ell}m_i=O(m)$. By Claim \ref{movement}(i), all implementations of Step \ref{nonstable} perform $O(\sum_{i=1}^{\ell}(m_i+n))=O(m+{\ell}n)=O(mn)$ moves, and thus, by Claim \ref{timen}, take $O(mn^2)$ time. Clearly all implementations of Steps \ref{uns}--\ref{reset} finish in $O({\ell}n^2)=O(mn^2)$ time. The overall implementation time of other steps  is  $O(mn)$.
\end{proof}

\paragraph{Performance.} Let $r^*\ge1$ denote the final common value of $r$ and $s$ in Algorithm~\ref{al5}. It is easy to see that the algorithm  implements Step \ref{nonstable} at least twice (as otherwise, $r^*=0$). Let $W_{\ell-1}$ (resp. $W_{\ell}$)  denote the  second-last (resp. last) implementation of (the while-loop at) Step \ref{nonstable}. Let $[K_1,K_2]$ and $[L_1,L_2]$  denote the cuts $[S_1,S_2]$ at the end of  $W_{\ell-1}$ and $W_{\ell}$, respectively.   It follows that both $[K_1,K_2]$ and $[L_1,L_2]$ are stable, and $r^*\le|[K_1,K_2]|\le|[L_1,L_2]|=r^*$. Therefore $[K_1,K_2]|=|[L_1,L_2]|=r^*$, implying that between  $W_{\ell-1}$ and $W_{\ell}$, no implementation of Step \ref{reset}
increases the cut size. After $W_{\ell-1}$, the algorithm does not change the cut $[S_1,S_2]=[K_1,K_2]$ (i.e., it schedules all nodes of $K_1$ with $\mathcal Y$, and all nodes of $K_2$ with $\mathcal N$) until it swaps $S_1$ and $S_2$ at Step \ref{swp}. Subsequently, $W_\ell$ starts with
\begin{gather}\label{start}
[S_1,S_2]=[K_2,K_1],\text{ where }[S_2,S_1]=[K_1,K_2]\text{ is stable.}
\end{gather}
  Since $W_{\ell}$ does not increase the cut size, any violating node $v$ satisfies  $d_{S_1}(v)=d_{S_2}(v)$ at the time it is moved. Therefore, recalling (\ref{stable}), the moves conducted by $W_\ell$ (if any) are  type-2 ones, which  move  nodes  from $S_2$ to $S_1$. Let $T$ ($\subseteq K_1$) denote the set of all these nodes moved.
It is clear that $V$ is the disjoint union of $K_2$, $L_2$ and $T$ such that
\begin{gather}
\text{$[K_2,K_1]=[K_2,L_2\cup T]$ and $[L_1,L_2]=[K_2\cup T, L_2]$.}\label{switch}
\end{gather}
 After $W_\ell$, the algorithm schedules all nodes of $L_1$ with $\mathcal Y$ and all nodes of $L_2$ with $\mathcal N$, finishing the last run of the inner repeat-loop.

\begin{clm}\label{transfer}
{If $T\ne\emptyset$}, then $T$ is an independent  set of graph $G$, and $d_{L_1}(v)=d_{L_2}(v)\ge1$ holds   for every $v\in T$.
\end{clm}
\begin{proof}
Suppose on the contrary that two nodes $x,y\in T$ are adjacent, and the {while-loop $W_\ell$ moves} $x$   earlier than  moving $y$ (from $S_2$ to $S_1$).  {By (\ref{start}), at the beginning of $W_\ell$, cut} $[S_2,S_1]$ is   stable. Therefore    $d_{S_2}(v)\le d_{S_1}(v)$ holds for all $v\in S_2$ at any time of  this while-loop. At the time {$W_\ell$} considers $y$, node $x$ has been moved to $S_1$ and  $d_{S_2}(y)=d_{S_1}(y)$. The adjacency of $x$ and $y$ implies that $d_{S_2}(y)>d_{S_1}(y)$ and $y\in S_2$ hold before $x$ is removed from $S_2$, which is a contradiction. So $T$ is an independent set. It follows  that throughout the while-loop,   $d_{S_1}(v)=d_{S_2}(v)$ holds for any $v\in T$. Moreover, $ d_{S_2}(v)\ge1$   for any $v\in T$ follows from the fact that $G$ is connected, and $T$ is independent.
\end{proof}

\begin{theorem}\label{per}
Algorithm \ref{al5} finds   a regret-proof schedule of $G$   that ensures at least $\max\{\sqrt{n+1}~-1,\frac12{(n-\alpha)}\}$ decisions of $\mathcal N$, where $\alpha$ is the independence number of $G$.
\end{theorem}
\begin{proof}
Note that the schedule output by the algorithm has its associated cut stable. Thus the algorithm does output a regret-proof schedule. {Suppose the output schedule ensures a number of $k$ decisions of $\mathcal N$. Since the algorithm has scheduled all nodes of $K_2$ (resp. $L_2$) with $\mathcal N$, Step \ref{record} guarantees that \[k\ge\max\{|K_2|,|L_2|\}\ge(|V|-|T|)/2\ge(|V|-\alpha)/2,\] as $V$ is the disjoint union of $K_2,L_2,T$, and $T$ is either empty or an independent set of $G$. It remains to prove    $k\ge\lambda:=\sqrt{n+1}-1$.}

 {Suppose on the contrary that $k<\lambda$, saying $|K_2|<\lambda$ and $|L_2|<\lambda$.}
 It follows that $|[K_2,L_2]|\le|K_2|\cdot|L_2|\le \lambda^2$ and $|T|=|V\setminus K_1\setminus L_2|=|V|-|K_2|-|L_2|>n-2\lambda$. Notice from Claim \ref{transfer} that  {each node of $T$ is adjacent to at least one node of $L_2$, implying} $[L_2,T]\ge|T|>n-2\lambda$. By {(\ref{start}) and (\ref{switch})}, the stability of $[K_1,K_2]=[L_2\cup T, K_2]$ implies that   {every node $v\in L_2$ satisfies   $    d_{K_2}(v)\ge d_{K_1}(v)\ge d_T(v)$, giving $ [K_2,L_2 ]\ge [T,L_2]$. Hence $\lambda^2>n-2\lambda$, implying $\lambda>\sqrt{n+1}-1$}, a contradiction.
\end{proof}

\section{Conclusion}
 In this paper, we have studied, from an algorithmic point of view, the marketing schedule problem for promoting products with negative externalities, aiming at profit maximization (from the seller's perspective) and regret-free decisions  (from the consumers' perspective). We have shown that the problem of finding a schedule with maximum profit is NP-hard and admits constant approximation. We find in strongly polynomial time  schedules that lead to regret-free decisions. These regret-proof schedules have satisfactory performance in terms of profit maximization, while it is left open whether both regret-proof-ness and constant profit approximation can be guaranteed in case of product $\mathcal N$ being more profitable.

 Our model and results apply to marketing one or two (types of) products with negative externalities in undirected social networks.  An interesting question is what happens when marketing three or more (types of) products and/or the network is directed.

\paragraph{Acknowledgments.}  The authors are indebted to Professor Xiaodong Hu for stimulating and helpful discussions.

 \bigskip
\appendix
\section*{APPENDIX}
\section{Proof of Theorem \ref{maxN}} \label{ap:a}
  We prove the NP-hardness of maximizing the number of $\mathcal N$ decisions 
 by
reduction from the {\em 3-OCC-MAX-2SAT} problem. It is a restriction of the MAX-2SAT problem, which, given a collection of disjunctive clauses of literals, each clause having at most
two literals, and each literal occurring in at most three clauses, is to find a truth assignment to satisfy as many clauses
as possible. It is known that  3-OCC-MAX-2SAT is NP-hard \cite{py91}.


\paragraph{Construction.}   Consider  any instance $I$ of the 3-OCC-MAX-2SAT problem:   $N$ boolean variables $x_1,x_2,\ldots,$ $x_N$  and $M$ clauses $y_j=(x_{j1}\vee x_{j2})$, $j=1,2,\ldots,M$,
where {$x_{j1},x_{j2}\in \{x_1,x_2,\cdots,x_N,\neg x_1,\neg x_2,\cdots, \neg x_N\}$, $j=1,2,\ldots,M$.} We construct an instance   $G=(V,E)$ of the
rebel scheduling problem in polynomial time 
 as follows.

  \begin{itemize}
  \item Create a pair of {\em literal nodes} $x_i$ and $\neg x_i$ representing, respectively, variable $x_i$ and its negation, $i=1,2,\ldots,N$;
    \item Create a {\em clause node} $y_j$ representing clause $y_j$, $j=1,2,\ldots,M$;
    \item Link literal node $x$ with clause node $y$ iff literal $x$ occurs in clause $y$;
    \item Create a {\em gadget} $G_i$ for each pair of literal nodes $x_i$ and $\neg x_i$, $i=1,2,\ldots,N$ (see Fig. \ref{f1}) as follows: let $L=10N+M$,
    \begin{itemize}
   \item  add four groups of nodes,
$A_i=\{a^i_1,a^i_2 \}$, $B_i=\{b^i_1,b^i_2\}$, $C_i=\{c^i_1,c^i_2,\ldots, c^i_{9}\}$, $D_i=\{d^i_{k1},d^i_{k2},\ldots, d^i_{kL}:k=1,2,\ldots,  9\}$;
\item link $x_i$ and $\neg x_i$ with all $13$ nodes in $A_i\cup B_i\cup C_i$;
\item link $b^i_1$ and $b^i_2$ with all nodes in $C_i$;
\item link $c^i_k$ with all nodes in $\{c^i_{k+1},d^i_{k1},d^i_{k2},\ldots, d^i_{kL}\}$ for $k=1,2,\ldots,9$, where $c^i_{10}=c^i_1$.
    \end{itemize}
  \end{itemize}
Clearly, $|V|=2N+M+N(13+9L)=M+(15+9L)N$. Clause nodes are not contained in any gadget. Each literal node is contained in a unique gadget $G_i$;
it has exactly   $13$ neighbors in $G_i$, and at most 3 neighbors outside $G_i$, which correspond to the clauses containing it. Each node in
$A_i$ has exactly two neighbors $x_i$ and $\neg x_i$. Each $C_i$ induces a cycle. All nodes in $D_i$ are pendant.

 {Let $opt(I)$ denote the optimal value for the 3-OCC-MAX-2SAT instance $I$. Let $opt(G)$ denote the maximum number of $\mathcal N$ decisions  contained in the outcome of a schedule for $G$. We will prove in Lemmas \ref{cl} and \ref{ok} that $opt(G)=opt(I)+ (5+9L )N$, which establishes Theorem \ref{maxN}. Under the optimality, we will show that the literal nodes with $\mathcal Y$ decisions in an optimal schedule correspond  to TRUE literals   in an optimal truth assignment. The gadget $G_i$ is used to guarantee that exactly one of $x_i$ and $\neg x_i$ chooses $\mathcal Y$.}


\begin{figure}
 \begin{center}
  \includegraphics[width=11cm]{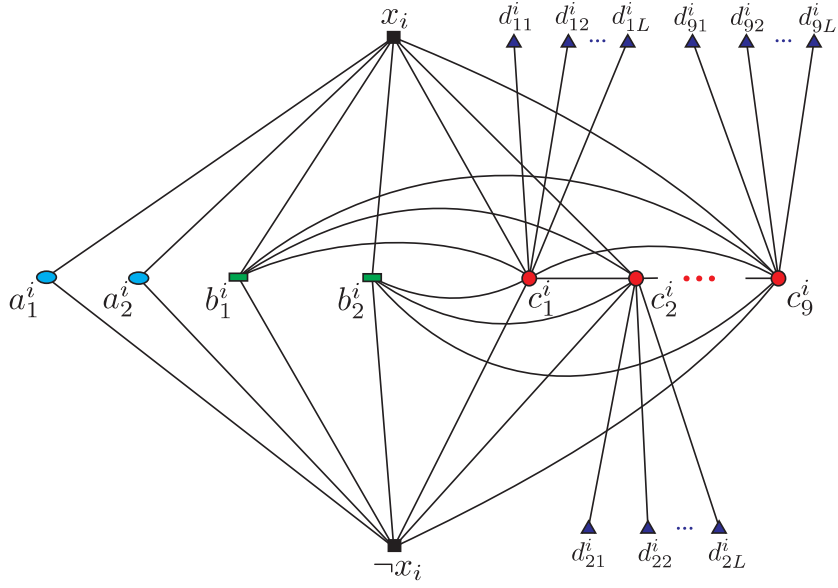}\\
  \end{center}
  \caption{Gadget $G_i$.}\label{f1}
\end{figure}

\paragraph{Schedule.} We construct a schedule $\pi$ for $G$ under which $opt(I)$ clause nodes, $N$ literal nodes and all $N(4+9L)$ nodes in $\cup_{i=1}^NA_i\cup B_i\cup D_i$ choose $\mathcal N$, which proves the following lemma.

\begin{lemma}  $opt(G)\geq opt(I)+(5+9L )N.$ \label{cl}
\end{lemma}

\begin{proof}

{Let $T$ be the set of  TRUE literals in  an optimal truth assignment   of $I$. Then $T$ is an independent set of literal nodes in $G$ such that} 
for each $i=1,2,\ldots,N$, exactly one of $x_i$ and  $\neg x_i$ is contained in $T$.
 The schedule $\pi$ proceeds in two stages.

 In the first stage, $\pi$
schedules the {(literal)} nodes in $T$ and then the {$M$} clause nodes. Since {$T$ is independent, all its nodes }
choose $\mathcal Y$. {Therefore, the $opt(I)$ clause nodes (which correspond to the satisfied clauses)} all choose $\mathcal N$. 

In the second stage, $\pi$ schedules {gadgets one after another}  in an arbitrary order.  For each gadget
$G_i$, {let $x$ be $x_i$ or $\neg x_i$ whichever belongs to $T$ and thus has chosen $\mathcal Y$.  Within subnetwork $G_i$, schedule $\pi$ proceeds in} 
five steps. (1) $\pi$ schedules {the nodes in the independent set $A_i\cup B_i$ first; obviously these nodes} 
all choose
$\mathcal N$ {due to their common neighbor $x$}. (2) Then $\pi$ schedules $c^i_1$, $c^i_2$, $\ldots$, $c^i_{8}$ {in this order. When $c^i_1$ is scheduled, she has exactly one neighbor choosing $\mathcal Y$, i.e., $x$, and two neighbors choosing $\mathcal N$, i.e., $b^i_1,b^i_2$. Therefore $c^i_1$ chooses $\mathcal Y$. Inductively, for $k=2,3,\ldots,8$, given $\pi[c^i_{k-1}]=\mathcal Y$, when $\pi$ schedules $c^i_k$, the node $c^i_k$ has exactly two neighbors choosing $\mathcal Y$ (i.e., $x,c^i_{k-1}$) and exactly two neighbors choosing $\mathcal N$ (i.e., $b^i_1,b^i_2$), which implies $\pi[c^i_k]=\mathcal Y$.}
(3) Next, $\pi$ schedules {$\neg x$. At that time, inside $G_i$ node $\neg x$ has exactly $|A_i\cup B_i|=4$ neighbors   choosing $\mathcal N$ and $|C_i|-1=8$ neighbors choosing $\mathcal Y$;
outside $G_i$, node $\neg x$ has at most 3 neighbors.  It follows that $\neg x$ chooses $\mathcal N$. (4) Now $\pi$ schedules $c^i_{9}$. At this time $c^i_{9}$ has exactly three neighbors choosing $\mathcal Y$ (i.e., $x$, $c^i_{8}$, $c^i_1$) and exactly three neighbors choosing $\mathcal N$ (i.e., $b^i_1,b^i_2,\neg x$)}. Therefore $c^i_{9}$ chooses $\mathcal Y$ as all other nodes of $C_i$ do.
(5) In the last step, $\pi$ schedules {the nodes in $D_i$, all with decisions $\mathcal N$.}

 {Since $\pi$ schedules each $G_i$ with $|A_i|+|B_i|+|\{\neg x\}|+|D_i|=5+9L$ decisions of $\mathcal N$, it follows that $\pi$ schedules $G$ with $opt(I)+N(5+9L)$ decisions of $\mathcal N$, establishing the lemma.}
\end{proof}

\paragraph{Assignment.}  Let $\pi^*$ be a schedule for $G$ that leads to a maximum number $opt(G)$ of $\mathcal N$ decisions.   To establish the reverse inequality of the one in Lemma~\ref{cl}, we will construct a truth assignment for $I$ based on $\pi^*$'s schedule of literal node. Notice from Lemma \ref{cl}, $|V|=M+(15+9L)N$ and $L=10N+M $ that $opt(G)>N(5+9L)= |V|-L$.

  \begin{clm}\label{sy}
 $\pi^*$ schedules all $9N$ nodes in $\cup_{i=1}^NC_i$ with decision $\mathcal Y$, and therefore (by the   maximality of $opt(G)$) schedules all nodes in $D^i_k$ after $c_k^i$ with $\mathcal N$ decisions for any $i=1,2,\ldots,N$ and $k=1,2,\ldots,9$.
 \end{clm}
 \begin{proof}  If $\pi^*$  schedules some $c^i_k\in C_i$ with decision $\mathcal N$, then all the $L$ nodes $d^i_{k1},d^i_{k2},\ldots, d^i_{kL}$ choose $\mathcal Y$ under $\pi^*$, a contradiction to $opt(G)>|V|-L$.
 \end{proof}

 \begin{clm} {Let $T^*$ be the set of literal nodes who choose $\mathcal Y$ under $\pi^*$.} For each $i=1,2,\ldots,N$, at most one of $x_i$ and  $\neg x_i$ is contained in $T^*$.\label{ce}\end{clm}
 \begin{proof}
 {Suppose that $\pi^*$ schedules some literal node $x\in\{x_i,\neg x_i\}$ with decision $\mathcal Y$ for some $i\in\{1,2,\ldots,N\}$.
  Note that $x$  has at most $16$ neighbors; 9} of them belong to $C_i$ and are scheduled by $\pi^*$ with decisions $\mathcal Y$ (see Claim \ref{sy}).
  It must be the case that $\pi^*$ schedules $x$ before the last scheduled node $c\in C_i$. At the time
$\pi^*$ schedules $c$,  by Claim \ref{sy}, $c$ has exactly two neighbors in $C_i$ choosing $\mathcal Y$, and has no neighbor in
$D_i$ that has been scheduled. The other four neighbors of $c$ are $x,\neg x,b_1,b_2$. It follows from $\pi^*[c]=\mathcal Y$ and $\pi^*[x]=\mathcal Y$ that $\pi^*$ schedules $\neg x$, $b_1$ and $b_2$ before $c$ with decision $\mathcal N$.
 \end{proof}

\begin{lemma} $opt(G)\leq opt(I)+ (5+9L)N$.
\label{ok}\end{lemma}
\begin{proof}  {By Claim \ref{ce}, $\{1,2,\ldots,N\}$ is the disjoint union of two sets $K_1$ and $ K_2$ such that $\pi^*$ schedules  exactly one of $x_i$ and $\neg x_i$ with $\mathcal N$ for every $i\in K_1$, and schedules $x_i$ and $\neg x_i$ with $\mathcal N$ for every $i\in   K_2$. Note that $|K_1|=|T^*|$, $|T^*|+|K_2|=N$ and $\pi^*$ schedules all nodes in $\{a_1^i,a_2^i:i\in K_2\}$  with $\mathcal Y$.}

{Define a truth assignment of $I$ by setting a literal to be TRUE if and only if it belongs to $ T^* \cup\{x_i:i\in K_2\}$. Note that any clause node $y$ with decision $\pi^*[y]=\mathcal N$  must have a neighbor (which is a literal node) choosing $\mathcal Y$ under $\pi^*$. This neighbor thus belongs to $T^*$. Thus clause $y$ is satisfied by the truth assignment. It follows that $\pi^*$ schedules at most $opt(I)$ clause nodes with $\mathcal N$.
From Claim \ref{sy} we deduce that
\begin{eqnarray*}
opt(G)&\le&|V\setminus\cup_{j=1}^N C_j\setminus T^*\setminus\{a_1^i,a_2^i:i\in K_2\}\setminus\{y_j:j=1,2,\ldots,M\}|+opt(I)\\
&=&|V|-9N-|T^*|-2|K_2|-M+opt(I)\\
&=&|V|-10N-|K_2|-M+opt(I).
 \end{eqnarray*}
 It follows from $|V|\!=\!M+(15+9L)N$ that $opt(G)\!\le\!(5+9L)N+opt(I)-|K_2|$.} \end{proof}

\section{Time complexity in Theorem  \ref{th4}}\label{ap:b}
\paragraph{Preprocessing.} Initially, we set graph $H $ to be $G=(V,E)$, We find in $O(m)$ time a maximal independent set $X$ of $H$, and set $Y:=V\setminus X$.

To find $X_i,Y_i$, $i=0,1,2,\ldots,\ell$, we will modify $H$ step by step via removing some nodes (together with their incident links). At any step,
we call a node of $H$ an $X$-node (resp. a $Y$ node) if this node belongs to $X$ (resp. $Y$). In $H$, a $Y$-node is {\em critical} if it is adjacent
to a pendant $X$-node. Any single non-critical node can be removed from $H$ without destroying the $Y$-node domination of $X$-nodes.

Inductively, we consider $i=0,1,\ldots,\ell$ in this order. The $i$-th stage of the process starts with $H=G_i$ and $Y_i=\emptyset$. Subsequently,
\begin{itemize}
\item[(i)] whenever $H$ has a non-critical $Y$-node $v$, we remove $v$ from $H$, add $v$ to $Y_i$, and update $H$.
\end{itemize}
The repetition finishes when all $Y$-nodes in $H$ are critical. At that time, the $i$-th stage finishes with
\begin{itemize}
\item[(ii)] outputting $Y_i$ and $X_{i+1}$ the set of pendant nodes;
\item[(iii)] removing all nodes of $X_{i+1}$ from $H$, and updates $H$ which gives $G_{i+1}=H$.
\end{itemize}

\paragraph{Running time.} Next we show that all the above $\ell+1$ stages finish in $O(m)$ time. At the initiation step, in $O(m)$ time we  find the set of pendant $X$-nodes, and the set of  non-critical  $Y$-nodes of $H$, where $H=G$.

As our preprocessing proceeds, when we remove a $Y$-node $v$ from $H$, we update $H$ by modifying the adjacency list representation of $H$, and
 \begin{itemize}
 \item updating the degrees of all $X$-nodes;
 \item updating the set of pendant $X$-nodes (using the degrees updated);
 \item updating the set of non-critical $Y$-nodes (using the   pendant $X$-nodes created).
 \end{itemize}
These can be done in $O(d(v))$ time, where $d(v)$ is the degree of $v$ in $G$. Thus in the whole process, the removals of $Y$-nodes and their corresponding update
 in (i) take $O(\sum_{v\in Y}d(v))=O(|E|)=O(m)$ time.

When we remove all pendant $X$-nodes from $H$, we update by  modifying the adjacency list representation of $H$,
 \begin{itemize}
 \item updating the set of pendant $X$-nodes (i.e., setting it to be empty);
 \item updating the set of non-critical $Y$-nodes (i.e., enlarging it by  the unique $Y$-neighbors of the removed $X$-nodes).
 \end{itemize}
 Hence throughout removals of pendant $X$-nodes and their corresponding update in (iii) takes  $O(|X|)=O(n)$ time.

 Since throughout the process, we have an updated set of non-critical $Y$-nodes at hand, at any time, finding a non-critical node of $H$ takes $O(1)$ time.
  The overall running time of (i) is $O(m)$, so is that of  (iii).  Note that all $X_i,Y_i$, $i=1,2,\ldots,\ell$, are mutually disjoint. Hence, overall, (ii) takes $O(n)$ time. Recall that $X_0=\emptyset$.
  We have the following   results.

 \begin{lemma}
 All $X_i$ and $Y_i$, $i=0,1,\ldots,\ell$ can be found in   $O(m)$  time.\qed
 \end{lemma}

 Since $G_i\subseteq G_{i-1}$ for $i=\ell,\ell-1,\ldots,1$, the refinement of Algorithm \ref{al1}, stated in  Algorithm \ref{al2}, runs $O(n^2)$ time.

\end{document}